\documentclass[reprint, longbibliography, noeprint]{revtex4-1}

\usepackage[utf8]{inputenc}

\usepackage[lining,semibold]{libertine} 
\usepackage[libertine, cmintegrals, bigdelims, vvarbb]{newtxmath}

\usepackage{amsmath}
\usepackage{amsfonts}
\usepackage{mathrsfs}
\usepackage{gensymb}
\usepackage{bbm}
\usepackage{dsfont}

\usepackage{kbordermatrix}
\renewcommand{\kbldelim}{(}
\renewcommand{\kbrdelim}{)}

\usepackage{chemformula}
\usepackage[caption=false]{subfig}
\usepackage{chemfig}
\usepackage[version=3]{mhchem}

\usepackage{tikz}
\usetikzlibrary{matrix,positioning,decorations.pathreplacing}

\usepackage{soul}
\usepackage{xcolor}

\usepackage{scalerel} 
\usepackage{comment}

\definecolor{webgreen}{rgb}{0,.5,0}
\definecolor{webbrown}{rgb}{.6,0,0}
\definecolor{grigio}{rgb}{.85,.85,.85} 
\definecolor{RoyalBlue}{rgb}{0.0, 0.14, 0.4}
\definecolor{skyblue1}{rgb}{0.45,0.62,0.81}
\definecolor{skyblue2}{rgb}{0.2,0.39,0.64}
\definecolor{skyblue3}{rgb}{0.13,0.29,0.53}
\definecolor{scarlet1}{rgb}{0.93,0.16,0.16}
\definecolor{scarlet2}{rgb}{0.8,0,0}
\definecolor{scarlet3}{rgb}{0.64,0,0}

\definecolor{g}{gray}{0.50}

\usepackage{hyperref}
\hypersetup{%
    colorlinks=true, linktocpage=true, pdfstartpage=1, pdfstartview=FitV,%
    breaklinks=true, pdfpagemode=UseNone, pageanchor=true, pdfpagemode=UseOutlines,%
    plainpages=false, bookmarksnumbered, bookmarksopen=true, bookmarksopenlevel=1,%
    hypertexnames=true, pdfhighlight=/O,
    urlcolor=webbrown, linkcolor=RoyalBlue, citecolor=webgreen, 
    pdftitle={},%
    pdfauthor={Francesco Avanzini},%
    pdfsubject={},%
    pdfkeywords={},%
    pdfcreator={pdfLaTeX},%
    pdfproducer={LaTeX REVTeX}%
}
\newenvironment{theorem}{\textbf{Theorem.}}{}
\newenvironment{proof}{\textit{\textbf{Proof.}}}{\hfill$\square$}

\newcommand{\rd}[1]{\overline{#1}}

\newcommand{\dt}{\mathrm d_t}

\newcommand{\dk}{{\mathbb 1}}

\newcommand\flowsfromb{\mathrel{\reflectbox{$\leadsto$}}}

\newcommand{\iii}{i}

\newcommand{\vvve}{v}
\newcommand{\wwwe}{w}
\newcommand{\vvv}{{\boldsymbol \vvve}}
\newcommand{\www}{{\boldsymbol \wwwe}}

\newcommand{\spe}{\alpha}
\newcommand{\rct}{\rho}
\newcommand{\rcts}{\rho_{\spe}}
\newcommand{\com}{\kappa}

\newcommand{\nspe}{s}
\newcommand{\nspeq}{\nspe_Q}
\newcommand{\nspep}{\nspe_P}
\newcommand{\nrct}{r}
\newcommand{\nrctin}{\nrct_{\text{in}}}

\newcommand{\nspeqex}{\nrct_{Q_\text{ex}}}
\newcommand{\nspepex}{\nrct_{P_\text{ex}}}

\newcommand{\ncom}{{c}}

\newcommand{\setspe}{\mathcal S}
\newcommand{\setrct}{\mathcal R}
\newcommand{\setcom}{\mathcal C}

\newcommand{\setspeP}{\mathcal S_P}
\newcommand{\setspeQ}{\mathcal S_Q}
\newcommand{\setrctIn}{\mathcal R_{\text{in}}}
\newcommand{\setrctEx}{\mathcal R_{\text{ex}}}
\newcommand{\setrctExQ}{\mathcal R_{\text{ex}}^Q}
\newcommand{\setrctExP}{\mathcal R_{\text{ex}}^P}

\newcommand{\chem}{X_\spe}
\newcommand{\comp}[1]{Y_{\com#1}}
\newcommand{\comps}[1]{Y_{#1}}

\newcommand{\stoc}[1]{\nu_{\spe, #1\rct}}
\newcommand{\stov}[1]{\boldsymbol{\nu}_{#1\rct}}
\newcommand{\stovc}[1]{\boldsymbol{\nu}_{#1\com}}
\newcommand{\stovcf}[1]{\boldsymbol{\nu}_{#1\com(\rct)}}

\newcommand{\matS}{\nabla}
\newcommand{\colS}{\boldsymbol{\nabla}_{\rct}}
\newcommand{\colSs}[1]{\boldsymbol{\nabla}_{\rct_{#1}}}
\newcommand{\colSm}{\boldsymbol{\nabla}_{-\rct}}
\newcommand{\matI}{\partial}
\newcommand{\matIe}{\partial_{\com,\rct}}
\newcommand{\matC}{\Gamma}
\newcommand{\matCe}{\Gamma_{\spe,\com(\rct)}}

\newcommand{\matSin}{\nabla{}_{\text{in}}}
\newcommand{\matSinP}{\nabla{}^P_{\text{in}}}
\newcommand{\matSinQ}{\nabla{}^Q_{\text{in}}}
\newcommand{\matSex}{\nabla{}_{\text{ex}}}
\newcommand{\matSexP}{\nabla{}^P_{\text{ex}}}
\newcommand{\matSexQ}{\nabla{}^Q_{\text{ex}}}

\newcommand{\matIexP}{\partial^P_\text{ex}}

\newcommand{\matOinP}{\mathbb{O}^P_\text{in}}
\newcommand{\matOexQ}{\mathbb{O}^Q_\text{ex}}
\newcommand{\matOexP}{\mathbb{O}^P_\text{ex}}

\newcommand{\deff}{\delta}

\newcommand{\n}{\boldsymbol n}
\newcommand{\nex}{n}
\newcommand{\na}{\boldsymbol n_0}
\newcommand{\ns}{n_\spe}
\newcommand{\np}{\n + \colS}

\newcommand{\cc}{\boldsymbol{x}}
\newcommand{\sscc}{\boldsymbol{x}_{\mathrm{ss}}}

\newcommand{\stsp}{ \Omega(\na)}

\newcommand{\pnorm}{\mathcal N(\na)}

\newcommand{\dcurre}{u}
\newcommand{\dcurr}{\boldsymbol \dcurre}
\newcommand{\rdcurre}{\rd{u}}

\newcommand{\prob}{p_t}

\newcommand{\probss}{p_{\mathrm{ss}}}
\newcommand{\probeq}{p_{\mathrm{eq}}}
\newcommand{\probtrj}{\mathcal P}

\newcommand{\rprobtrj}{\rd{\mathcal P}}

\newcommand{\rate}[1]{\omega_{#1\rct}}
\newcommand{\ratex}[1]{\omega_{#1}}
\newcommand{\rrate}[1]{\rd{\omega}_{#1\rct}}
\newcommand{\rratex}[1]{\rd{\omega}_{#1}}
\newcommand{\currss}[1]{J{}_{#1\rct}^{\mathrm{ss}}}
\newcommand{\curreq}[1]{J{}_{#1\rct}^{\mathrm{eq}}}
\newcommand{\rcurrss}[1]{\rd{J}{}_{#1\rct}^{\mathrm{ss}}}

\newcommand{\indext}{l}
\newcommand{\nstep}{ L}

\newcommand{\nt}{\boldsymbol n}

\newcommand{\trj}[1]{\na\leadsto\n_{#1}}
\newcommand{\rtrj}[1]{\na\flowsfromb\n_{#1}}

\newcommand{\nl}{\n_\indext}
\newcommand{\nlm}{\n_{\indext-1}}
\newcommand{\nf}{\n}
\newcommand{\nfg}{\n_\nstep}

\newcommand{\rcttrja}{\rct_1}
\newcommand{\rcttrj}{\rct_\indext}
\newcommand{\rcttrjf}{\rct_\nstep}

\newcommand{\ttrja}{t_1}
\newcommand{\ttrj}{t_\indext}
\newcommand{\ttrjp}{t_{\indext+1}}
\newcommand{\ttrjf}{t_\nstep}

\newcommand{\ratel}[1]{\omega_{#1\rct_\indext}}
\newcommand{\rratel}[1]{\rd{\omega}_{#1\rct_\indext}}

\newcommand{\cycle}{\boldsymbol \phi}
\newcommand{\cyclee}{ \phi_\rct}
\newcommand{\cycleIn}{\cycle{}_{\text{in}}^{}}
\newcommand{\cycleEx}{\cycle{}_{\text{ex}}^{}}

\newcommand{\cycleExQ}{\cycle{}_{\text{ex}}^Q}
\newcommand{\cycleExP}{\cycle{}_{\text{ex}}^P}

\newcommand{\srct}{\boldsymbol \psi}
\newcommand{\srcte}{\psi_{\rct}}
\newcommand{\srctc}{\boldsymbol \psi_\text{cat}}

\newcommand{\naturaln}{{\mathbb N}}
\newcommand{\integer}{{\mathbb Z}}
\newcommand{\realp}{\mathbb R_{>0}}
\newcommand{\naturals}{\naturaln^\nspe}
\newcommand{\integerr}{\integer^\nrct}

\newcommand{\hr}[1]{h_{#1\rct}}
\newcommand{\h}{\boldsymbol{h}}
\newcommand{\generator}{\boldsymbol{b}}

\newcommand{\spex}{A}

\newcommand{\cost}{a}

\newcommand{\Dual}{Dual}
\newcommand{\dual}{dual}

\makeatletter
\def\maketag@@@#1{\hbox{\m@th\normalfont\normalsize#1}}
\makeatother

\DeclareMathAlphabet{\mathpzc}{OT1}{pzc}{m}{it}

\begin{document}

\title{Deficiency, Kinetic Invertibility, and Catalysis \\
in Stochastic Chemical Reaction Networks 
}

\newcommand\unilu{\affiliation{Complex Systems and Statistical Mechanics, Department of Physics and Materials Science, University of Luxembourg, L-1511 Luxembourg City, Luxembourg}}
\author{Shesha Gopal Marehalli Srinivas}
\email{shesha.marehalli@uni.lu}
\unilu
\author{Matteo Polettini}
\email{matteo.polettini@uni.lu}
\unilu
\author{Massimiliano Esposito}
\email{massimiliano.esposito@uni.lu}
\unilu
\author{Francesco Avanzini}
\email{francesco.avanzini@uni.lu}
\unilu


\date{\today}

\begin{abstract}
Stochastic chemical processes are described by the chemical master equation satisfying the law of mass-action. 
We first ask whether the \dual\ master equation, 
which has the same steady state as the chemical master equation, but with inverted reaction currents,
satisfies the law of mass-action, namely, still describes a chemical process. 
We prove that the answer depends on the topological property of the underlying chemical reaction network known as deficiency.
The answer is yes only for deficiency-zero networks. 
It is no for all other networks, implying that their steady-state currents cannot be inverted by controlling the kinetic constants of the reactions.
Hence, the network deficiency imposes a form of non-invertibility to the chemical dynamics.
We then ask whether catalytic chemical networks are deficiency-zero.
We prove that the answer is no when they are driven out of equilibrium due to the exchange of some species with the environment.
\end{abstract}

\maketitle


\section{Introduction}
Open chemical reaction networks (CRNs) driven out of equilibrium constitute the underlying mechanism of many complex processes in biosystems~\cite{Yang2021} and synthetic chemistry~\cite{Hermans2017}.
Information processing~\cite{Kholodenko2006}, oscillations~\cite{novak2008}, self-replication~\cite{Segre2001,Bissette2013,Lancet2018}, self-assembly~\cite{rossum2017,ragazzon2018} and molecular machines~\cite{zerbetto2007,Cakmak2015} provide some prototypical examples.
At steady state, these CRNs operate with nonzero net reaction currents sustained by 
thermodynamic forces generated via continuous exchanges of free energy with the environment~\cite{Esposito2020}.
Their energetics can be characterized on rigorous grounds using nonequilibrium thermodynamics of CRNs 
undergoing stochastic~\cite{gaspard2004, Schmiedl2007, Rao2018b} or deterministic dynamics~\cite{Qian2005,rao2016, avanzini2021, avanzini2022}.
For instance, this theory has been used to quantify 
the energetic cost of maintaining coherent oscillations~\cite{Oberreiter2022b};
the efficiency of dissipative self-assembly~\cite{penocchio2019eff} and central metabolism~\cite{Wachtel2022};
the internal free energy transduction of a model of chemically-driven self-assembly and an experimental light-driven bimolecular motor~\cite{Penocchio2022}.
It has also been used to determine speed limits for the chemical dynamics~\cite{Yoshimura2021a, Yoshimura2021b}.
Away from equilibrium, the currents are nonlinear functions of the thermodynamic forces.
While they increase with the forces close to equilibrium (i.e., for small forces), they can also decrease far from equilibrium (i.e., for large forces)~\cite{Altaner2015, falasco2019ndr}.
They can further show a highly nonlinear response to time-dependent modulations in the forces~\cite{Samoilov2002, forastiere2020} as well as lead to the emergence of oscillations or chaos~\cite{Fritz2020, Gaspard2020}.

In this paper,
we focus on the steady-state dynamics of stochastic chemical processes described by the chemical master equation satisfying the law of mass-action.
We start by investigating whether the net currents of all reactions can be inverted by controlling the kinetic constants only. 
To do so, we use the \dual\ master equation~\cite{Kemeny1976} (also called adjoint or reversal~\cite{Norris1997, Crooks2000, Chernyak2006}) of the chemical master equation which, by definition, has the same steady state but inverted currents.
It thus describes a stochastic process, called here \dual\ process, whose dynamics is inverted compared to the chemical process (see App.~\ref{app:trarep}).
From a thermodynamic perspective, 
the \dual\ process is not in general the time-reversed process which enters the definition of the entropy production of a stochastic trajectory,
but it enters the definition of the adiabatic and nonadiabatic entropy production~\cite{Esposito2010}.
From a dynamic perspective, 
the \dual\ process is not a chemical process unless the \dual\ master equation satisfies the law of mass-action.
By building on the results derived in Refs.~\cite{Anderson2010, Cappelletti2016} (and summarized in App.~\ref{sec:ssd0}),
we prove that this happens if and only if the underlying CRN has zero deficiency.
This constitutes our first main result.
The deficiency is a topological property of CRNs, which roughly speaking quantifies the number of ``hidden'' cycles 
(i.e., cycles that have not a graphical representation in the graph of complexes)~\cite{Polettini2015}.
Physically, our result means that the network deficiency determines the kinetic invertibility (or non-invertibility) of the stochastic chemical dynamics.
We further show that the correspondence between network deficiency and invertibility is specific of the stochastic dynamics:
in the thermodynamic limit, where the dynamics becomes deterministic, the net steady-state currents can always be inverted  independently of the network deficiency.

We then investigate which CRNs are not deficiency-zero. 
We consider catalytic CRNs which are ubiquitous in nature.
From a chemical point of view, a catalyst is a substance that acts as both a reactant and a product of the reaction while increasing its rate~\cite{iupac2009}.
From this definition, necessary stoichiometric conditions for catalysis in CRNs have been recently derived 
and used as a mathematical basis to identify minimal autocatalytic subnetworks (called motifs or cores) in larger CRNs~\cite{Blokhuis2020}.
By building on these results, we prove that catalytic CRNs are not deficiency-zero when they are driven out of equilibrium due to the exchange of some species with the environment.
This constitutes our second main result.

Together, our two main results show that the net steady-state currents of stochastic catalytic CRNs driven out of equilibrium via exchanges of species with the environment
cannot be inverted by controlling the kinetic constants.
We illustrate the problems studied in this paper in Fig.~\ref{fig:logic2}.
\begin{figure}[t]
  \centering
  \includegraphics[width=1.\columnwidth]{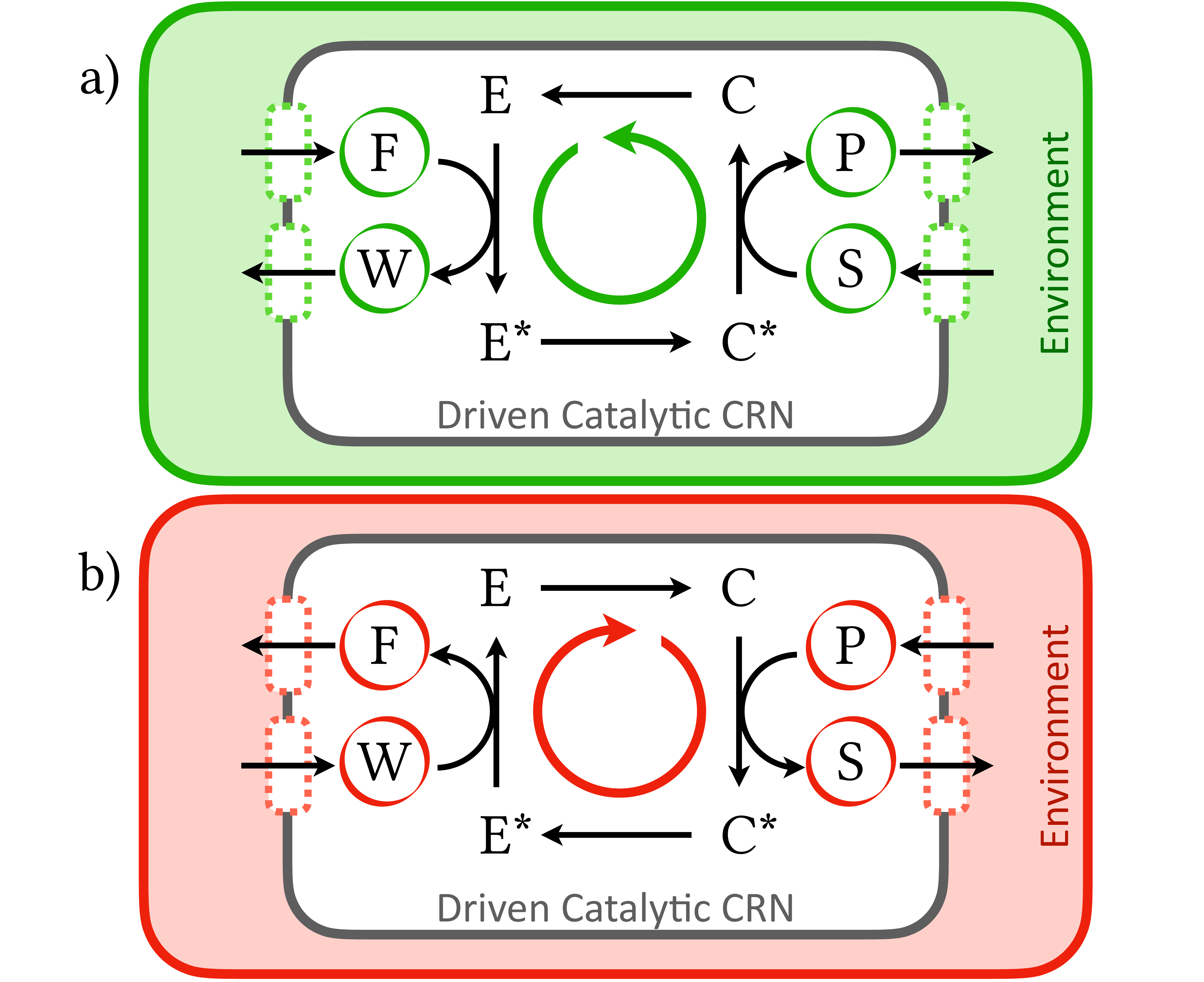}
\caption{
Illustration of a catalytic CRN driven out of equilibrium by exchange reactions (arrows crossing the boundary) with the environment.
Here, all reactions are assumed to be reversible even if they are represented with a hypergraph notation by single arrows specifying the direction of the net steady-state currents. 
The catalysts ($\ch{E}$, $\ch{E}^{*}$, $\ch{C}$, $\ch{C}^{*}$) promote the  
the interconversion of other species (fuel $\ch{F}$, waste $\ch{W}$, substrate $\ch{S}$, and product $\ch{P}$) that are exchanged with the environment.
\textbf{a)} The CRN is coupled to an environment such that the net steady-state currents of the internal reactions flow counter-clockwise.
Correspondingly, the environment continuously provides (resp. extracts) $\ch{F}$ and $\ch{S}$  (resp. $\ch{W}$ and $\ch{P}$).
\textbf{b)}  The CRN is coupled to another environment such that the chemical reactions have different kinetic constants and net steady-state currents (of both internal and exchange reactions) are inverted with respect to the case in \textbf{a)}.
In this paper, we prove that what is illustrated in \textbf{b)} cannot happen if what illustrated in \textbf{a)} happens.
Namely, there are no kinetic constants such that the net steady-state currents can be inverted in catalytic CRNs driven out of equilibrium by the exchange reactions.
}
\label{fig:logic2}
\end{figure}

Our work is organized as follows.
We start in Sec.~\ref{sec:crn} by introducing the basic setup.
In Subs.~\ref{sub:setup}, we define CRNs and their topological properties including deficiency.
In Subs.~\ref{sub:engine}, we formalize the notion of driven CRNs via exchange reactions.
In Sec.~\ref{sec:dynamics}, we examine the dynamics of chemical processes according to the chemical master equation following the law of mass-action (Subs.~\ref{sub:chem_process}) and 
we introduce their \dual\ processes described by the \dual\ master equation (Subs~\ref{sub:rev_chem_process}).
Our first main result is obtained in Sec.~\ref{sec:chem_rev}. 
Our second main result is obtained in Sec.~\ref{sec:deficiencycat}.
Finally, we discuss our results and their generalizations in Sec.~\ref{sec:discussion}.


\section{Chemical Reaction Networks\label{sec:crn}}

\subsection{Setup\label{sub:setup}}
In CRNs, chemical species $\{\chem\}$, identified by the indexes ${\spe\in\setspe=\{1,2,\dots,\nspe\}}$, are interconverted via chemical reactions $\rct\in\setrct=\{\pm1,\pm2,\dots,\pm\nrct\}$ 
\begin{equation}
\sum_{\spe} \stoc{}\, \chem\ch{<=>[ $\rct$ ][ $-\rct$ ]} \sum_{\spe} \stoc{-}\, \chem\,,
\label{eq:reaction}
\end{equation}
with $\stoc{}$ (resp. $\stoc{-}$) the stoichiometric coefficient of the species $\chem$ in reaction $\rct$ (resp. $-\rct$).
Chemical reactions are assumed here to be \textit{reversible}: for every (forward) reaction $\rct\in\setrct$, reaction $-\rct\in\setrct$ and denotes its backward counterpart.

An alternative representation of the chemical reactions~\eqref{eq:reaction} is given in terms of complexes $\{\comp{}\}$, identified by the indexes ${\com \in\setcom=\{1,2,\dots,\ncom\}}$, 
\begin{equation}
\comp{(\rct)} \ch{<=>[ $\rct$ ][ $-\rct$ ]} \comp{(-\rct)}\,,
\end{equation}
that are aggregates of species appearing as reactants in a reaction
\begin{equation}
\comp{(\rct)} = \sum_{\spe} \stoc{}\, \chem\,.
\end{equation}
Notice that different reactions might involve the same complex. 


The topology of CRNs is encoded in the stoichiometric matrix $\matS$
whose columns $\colS$ (for $\rct>0$) specify the net variation of the number of molecules for each species undergoing the $\rct$ reaction
and thus are given by 
\begin{equation}
\colS = \stov{-} - \stov{} \,,
\label{eq:matS}
\end{equation}
with $\stov{} = (\dots, \stoc{}, \dots)^{\intercal}_{\spe\in\setspe}$. 
By definition, $\colSm = -\colS$.
The stoichiometric matrix $\matS$ can also be written as the product between the composition matrix~$\matC$ and the incidence matrix~$\matI$:
\begin{equation}
\matS = \matC\matI\,.
\label{eq:matSdec}
\end{equation}
The former specifies the stoichiometric coefficient of each chemical species $\spe$ in each complex $\com$, i.e.,
\begin{equation}
\matCe = \stoc{}\,.
\end{equation}
The latter is the incidence matrix of the graph of complexes, 
namely, the graph obtained using complexes as nodes and reactions as edges.
Indeed, $\matIe$ specifies whether each complex $\com$ is a reagent (negative value) or a product (positive value) of each reaction~$\rct>0$, 
i.e.,
\begin{equation}
\matIe = \dk_{\com,\com(-\rct)} - \dk_{\com,\com(\rct)}\,,
\end{equation}
with $\dk_{i,j}$ the Kronecker delta.

Every (right) null eigenvector $\cycle$ of the stoichiometric matrix, named \textit{stoichiometric cycle},
\begin{equation}
\matS\cycle = 0\,,
\label{eq:scycle}
\end{equation}
denotes the number of times each reaction occurs along a transformation after which molecule numbers $\{\ns\}$ are restored to their initial values. Equation~\eqref{eq:matSdec} implies that any (right) null eigenvector of the incidence matrix $\matI$, named \textit{topological cycle}, is a stoichiometric cycle, but not vice-versa. 
This difference is encoded in the deficiency:
\begin{equation}
\deff = \mathrm{dim}\,\mathrm{ker} \matS - \mathrm{dim}\,\mathrm{ker}\matI\,,
\end{equation}
where $\mathrm{dim}\,\mathrm{ker}(\bullet)$ returns the dimension of the kernel of a matrix, 
i.e., the number of linearly-independent (right) null eigenvectors.
Stoichiometric cycles that are not topological cycles cannot be visualized as cycles on the graphical representation of the graph of complexes and are thus ``hidden''.

\subsection{Driven CRNs via Exchange Reactions~\label{sub:engine}}


Let us split, without loss of generality, the set of reactions~$\setrct$ into two disjoint subsets: $\setrct = \setrctIn \cup \setrctEx$.
The set~$\setrctEx$ includes the exchange reactions with the environment represented by the complex~$\emptyset$,
that are assumed here to be reactions of the type
\begin{equation}
\chem\ch{<=>[ ${\rcts}$ ][ ${-\rcts}$ ]} \emptyset\,,
\label{eq:exrct}
\end{equation}
as done for example in Refs.~\cite{Rao2018b,DalCengio2022}.
The environment is not further specified:
it can represent chemostats as in Ref.~\cite{Rao2018b} or a larger network of reactions the CRN is embedded into.
The set~$\setrctIn$ includes the other reactions, named internal reactions.
The stoichiometric matrix can thus be written as
\begin{equation}
\matS = (\matSin\,,\matSex) \,,
\label{eq:matSinex}
\end{equation}
where $\matSin$ (resp. $\matSex$) collects the columns $\colS$ of the stoichiometric matrix with $\setrctIn\ni\rct>0$ (resp. $\setrctEx\ni\rct>0$).
Note that every column $\setrctEx\ni\rcts>0$ of $\matSex$ has all null entries 
except for the row corresponding to $\chem$ that is equal to~$-1$ because of the stoichiometry of the exchange reactions~\eqref{eq:exrct}.

We now  say that CRNs are driven out of equilibrium via exchange reactions (hereafter called only driven CRNs for simplicity) when the following condition is satisfied.
The stoichiometric matrix~\eqref{eq:matSinex} admits at least one stoichiometric cycle~\eqref{eq:scycle}, which can always be written as
\begin{equation}
\cycle = \big(\cycleIn\,,\cycleEx\big)^\intercal
\label{eq:subcycle}
\end{equation}
by applying the splitting $\setrct=\setrctIn\cup\setrctEx$,
such that
\begin{equation}
\cycleEx \neq 0\,.
\label{eq:couplingEnv}
\end{equation}
The condition in Eq.~\eqref{eq:couplingEnv} also implies that $\cycleIn\neq0$ 
since $\matS\cycle=\matSin\cycleIn + \matSex\cycleEx=0$, but 
\begin{equation}
\matSex\,\cycleEx\neq0
\label{eq:cycleimplication}
\end{equation}
because each column of $\matSex$ has just one no-null entry which is equal to~$-1$.
This physically means that the cycle~\eqref{eq:subcycle} defines a transformation involving the exchange of matter with the environment (via the exchange reactions $\setrctEx\ni\rct>0$ such that $\cyclee\neq0$).

\paragraph*{Example.} 
Consider the following CRN
\begin{equation}
\begin{split}
\ch{F + $a$E&<=>[ ${+1}$ ][ ${-1}$ ] I }\\
\ch{I  &<=>[ ${+2}$ ][ ${-2}$ ] ($a$ + $b$)E + W}\\
\ch{E &<=>[ ${+3}$ ][ ${-3}$ ] $\emptyset$}\\
\ch{W &<=>[ ${+4}$ ][ ${-4}$ ] $\emptyset$}\\
\ch{F &<=>[ ${+5}$ ][ ${-5}$ ] $\emptyset$}\\
\end{split}\,,
\label{eq:excrn}
\end{equation}
where $a$ molecules of enzyme \ch{E} react with a fuel molecule \ch{F} producing the intermediate \ch{I} which decomposes into ${a+b}$ molecules of enzyme \ch{E} and one molecule of waste \ch{W} via the internal reactions $\setrctIn =\{\pm1\,,\pm2\}$.
The species \ch{E}, \ch{W} and \ch{F} are exchanged with the environment~$\emptyset$ via the exchange reactions $\setrctEx =\{\pm3\,,\pm4\,,\pm5\}$.
The stoichiometric matrix of the CRN~\eqref{eq:excrn} is specified by
\begin{equation}
\matS=
 \kbordermatrix{
    & \color{g}1 &\color{g}2&\color{g}3 &\color{g}4&\color{g}5\cr
    \color{g}\ch{I}    	& 1 &-1 	   & 0 & 0 & 0\cr
    \color{g}\ch{E}   	&-a & (a+b) &-1 & 0 & 0\cr
    \color{g}\ch{W}  	& 0 & 1 	   & 0 &-1 & 0\cr
    \color{g}\ch{F}   	&-1 & 0 	   & 0 & 0 &-1\cr
  }\,,
\label{eq:exmatS}
\end{equation}
which admits one stoichiometric cycle for every value of the stoichiometric coefficients $a\geq0$ and $b\geq0$: 
\begin{equation}
\cycle =(\underbrace{1,1}_{=\cycleIn},\underbrace{b,1,-1}_{=\cycleEx})^\intercal\,.
\label{eq:excycle}
\end{equation}
Since $\cycleEx\neq0$ for every value of the parameters~$a$ and~$b$, the CRN~\eqref{eq:excrn} is always driven.

When $a=b=0$, there are four complexes $\comps{1}=\ch{I}$, $\comps{2}=\ch{E}$, $\comps{3}=\ch{W}$, $\comps{4}=\ch{F}$, and  $\comps{5}=\emptyset$, and all reactions become unimolecular.
The incidence matrix reads
\begin{equation}
\matI=
 \kbordermatrix{
    & \color{g}1 &\color{g}2&\color{g}3&\color{g}4&\color{g}5\cr
    \color{g}\comps{1}   	& 1 &-1 & 0 & 0 & 0\cr
    \color{g}\comps{2} 	& 0 & 0 &-1 & 0 & 0\cr
    \color{g}\comps{3}		& 0 & 1 & 0 &-1 & 0\cr
    \color{g}\comps{4}		&-1 & 0 & 0 & 0 &-1\cr
    \color{g}\comps{5}		& 0 & 0 & 1 & 1 & 1\cr
  }\,,
\end{equation}
and admits one topological cycle $(1,1,0,1,-1)^\intercal$.
Thus, the CRN~\eqref{eq:excrn} has deficiency $\deff=0$.

For any other value of stoichiometric coefficients $a$ and $b$, the CRN~\eqref{eq:excrn} has no topological cycles and, therefore, $\deff=1$.
For instance, when $a=0$ and $b>0$, there are six complexes $\comps{1}=\ch{I}$, $\comps{2}=\ch{$b$E + W}$, $\comps{3}=\ch{E}$, $\comps{4}=\ch{W}$,  $\comps{5}=\ch{F}$, and $\comps{6}=\emptyset$, and the incidence matrix
\begin{equation}
\matI=
 \kbordermatrix{
    & \color{g}1 &\color{g}2&\color{g}3&\color{g}4&\color{g}5\cr
    \color{g}\comps{1}   	& 1 &-1 & 0 & 0 & 0\cr
    \color{g}\comps{2} 	& 0 & 1 & 0 & 0 & 0\cr
    \color{g}\comps{3}		& 0 & 0 &-1 & 0 & 0\cr
    \color{g}\comps{4}		& 0 & 0 & 0 &-1 & 0\cr
    \color{g}\comps{5}		&-1 & 0 & 0 & 0 &-1\cr
    \color{g}\comps{6}		& 0 & 0 & 1 & 1 & 1\cr
  }\,
\end{equation}
has no right null eigenvectors.
In Sec.~\ref{sec:deficiencycat} we will see that the CRN~\eqref{eq:excrn} represents a driven catalytic CRN when $a>0$ and $b\geq0$ and, consequently, 
has deficiency $\deff>0$.

Notice that the CRN~\eqref{eq:excrn} does not always have deficiency~$\deff>0$ despite always being driven.
Namely, being driven is not a sufficient condition to have deficiency~$\deff>0$.

\section{Stochastic Dynamics\label{sec:dynamics}}

\subsection{Chemical Process\label{sub:chem_process}}
Each reaction~\eqref{eq:reaction} is assumed here to be a stochastic event. 
The integer-valued population state $\n =(\dots,\ns,\dots)^{\intercal}_{\spe\in\setspe}$ specifying the number of molecules of each species is a fluctuating variable. 
Chemical processes are described in terms of Markov jump processes,
and the probability $\prob(\n)$ of there being $\n$ molecules at time $t$ 
follows the chemical master equation~\cite{McQuarrie1967,Gillespie1992}
\begin{equation}
\dt \prob(\n) = \sum_{\rct\in\setrct}\{\rate{-}(\np)  \prob(\np) - \rate{}(\n) \prob(\n) \}\,,
\label{eq:cme}
\end{equation}
where the reaction rates satisfy the law of mass-action~\footnote{
The law of mass-action is satisfied for elementary reactions in ideal dilute solutions, namely, chemical species are noninteracting and there is a species, called the solvent, which is not involved in the chemical reactions and is much more abundant than all the other species.
If reactions are not elementary or the chemical species interact, reaction rates might not satisfy the law of mass-action.
In this work, we consider only elementary reactions in ideal dilute solutions.
}:
\begin{equation}
\rate{}(\n) = k_{\rct}\frac{\n!}{(\n-\stov{})!}\,,
\label{eq:massaction}
\end{equation}
with $\n! = \prod_{\spe}\ns!$ and $ k_{\rct}>0$ being the kinetic rate constant of reaction $\rct$.
Note that if $\n<\stov{}$ reaction $\rct$ cannot occur and hence $\rate{}(\n) =0$.
We {assume} that the dynamics occurs on a so-called stoichiometric subspace $\stsp$, namely, the set of all populations connected to $\na$ by arbitrary sequences of reactions: $\stsp = \{ \na + \matS\srct\geq0\, |\, \forall \srct\in \integerr\}$.

We further assume that the chemical master equation~\eqref{eq:cme} admits a unique steady-state probability $\probss(\n)$ such that
\begin{equation}
\sum_{\rct\in\setrct}\{\rate{-}(\np)  \probss(\np) - \rate{}(\n) \probss(\n) \}=0\,.
\label{eq:ss}
\end{equation}
If $\stsp$ is finite, the existence of $\probss(\n)$ is granted by the Perron-Frobenius theorem.
In general, the net steady-state current of the reactions along each stoichiometric cycle~\eqref{eq:scycle},
namely, $\rct>0$ such that $\cyclee\neq 0$ for every $\cycle$, does not vanish~\cite{Schnakenberg1976}:
\begin{equation}
\currss{} (\n)= \rate{}(\n) \probss(\n) - \rate{-}(\np)  \probss(\np) \neq 0\,,
\end{equation}
which physically means that these reactions continuously occur in a preferential direction (either forward if $\currss{} (\n) > 0$ or backward if $\currss{} (\n) < 0$) despite the CRN being at steady state.

Only when CRNs have no stoichiometric cycles~\eqref{eq:scycle}
(or the kinetic constants have specific values satisfying the Wegscheider’s condition~\cite{Schuster1989} for every stoichiometric cycle, i.e., $\prod_{\rct>0}(k_{\rct}/k_{-\rct})^{\cyclee}=1$ $\forall \cycle$),
the steady-state probability is an equilibrium probability satisfying the detailed balance condition
\begin{equation}
\rate{}(\n) \probeq(\n) = \rate{-}(\np)  \probeq(\np)\,,
\label{eq:db}
\end{equation}
and $\curreq{} (\n) =0$ for every $\rct>0$.

\subsection{\Dual\ Process\label{sub:rev_chem_process}} 
The \dual\ master equation~\cite{Kemeny1976} (also called adjoint or reversal~\cite{Norris1997}) of the chemical master equation~\eqref{eq:cme} is endowed with rates
\begin{equation}
\rrate{}(\n) = \rate{-}(\np)\frac{\probss(\np)}{\probss(\n)}\,,
\label{eq:rrate}
\end{equation}
where $\rate{}(\n)$ and ${\probss(\n)}$ are given in Eqs.~\eqref{eq:massaction} and~\eqref{eq:ss}, respectively.
The resulting Markov jump process, called here \dual\ process, is not an abstract mathematical construction:
by definition, the \dual\ master equation has the same steady-state probability as the original chemical master equation, i.e.,
\begin{equation}
\sum_{\rct\in\setrct}\{\rrate{-}(\np)  \probss(\np) - \rrate{}(\n) \probss(\n) \}=0\,,
\label{eq:ssr}
\end{equation}
but with inverted net steady-state currents, i.e., 
\begin{equation}\small
\rcurrss{}(\n) = \rrate{}(\n) \probss(\n) - \rrate{-}(\np)  \probss(\np) = -\currss{}(\n)\,.
\end{equation}
Furthermore, at steady state the evolution of the population~$\n$ in terms of stochastic trajectories in the \dual\ process is inverted compared to the original chemical process as shown in App.~\ref{app:trarep}.
However, the \dual\ master equation does not describe in general chemical processes:
the \dual\ rates~\eqref{eq:rrate} do not satisfy the law of mass-action~\eqref{eq:massaction}.

\section{Kinetic Invertibility\label{sec:chem_rev}}
We investigate now when the \dual\ process is still a chemical process,
namely, when there exist kinetic constants~${\rd{k}_{\rct}>0}$ such that the \dual\ rates~\eqref{eq:rrate} satisfy the law of mass-action:
\begin{equation}
\rrate{}(\n) =  \rd{k}_{\rct}\frac{\n!}{(\n-\stov{})!}\,.
\label{eq:rratema}
\end{equation}
We call this condition kinetic invertibility.
In particular, we prove that Eq.~\eqref{eq:rratema} holds if and only if CRNs are deficiency-zero, 
besides the trivial case when the detailed balance condition~\eqref{eq:db} is satisfied (implying $\rrate{}(\n) = \rate{}(\n)$),
as summarized in Fig.~\ref{fig:reversed}.
\begin{figure}[t]
  \centering
  \includegraphics[width=0.99\columnwidth]{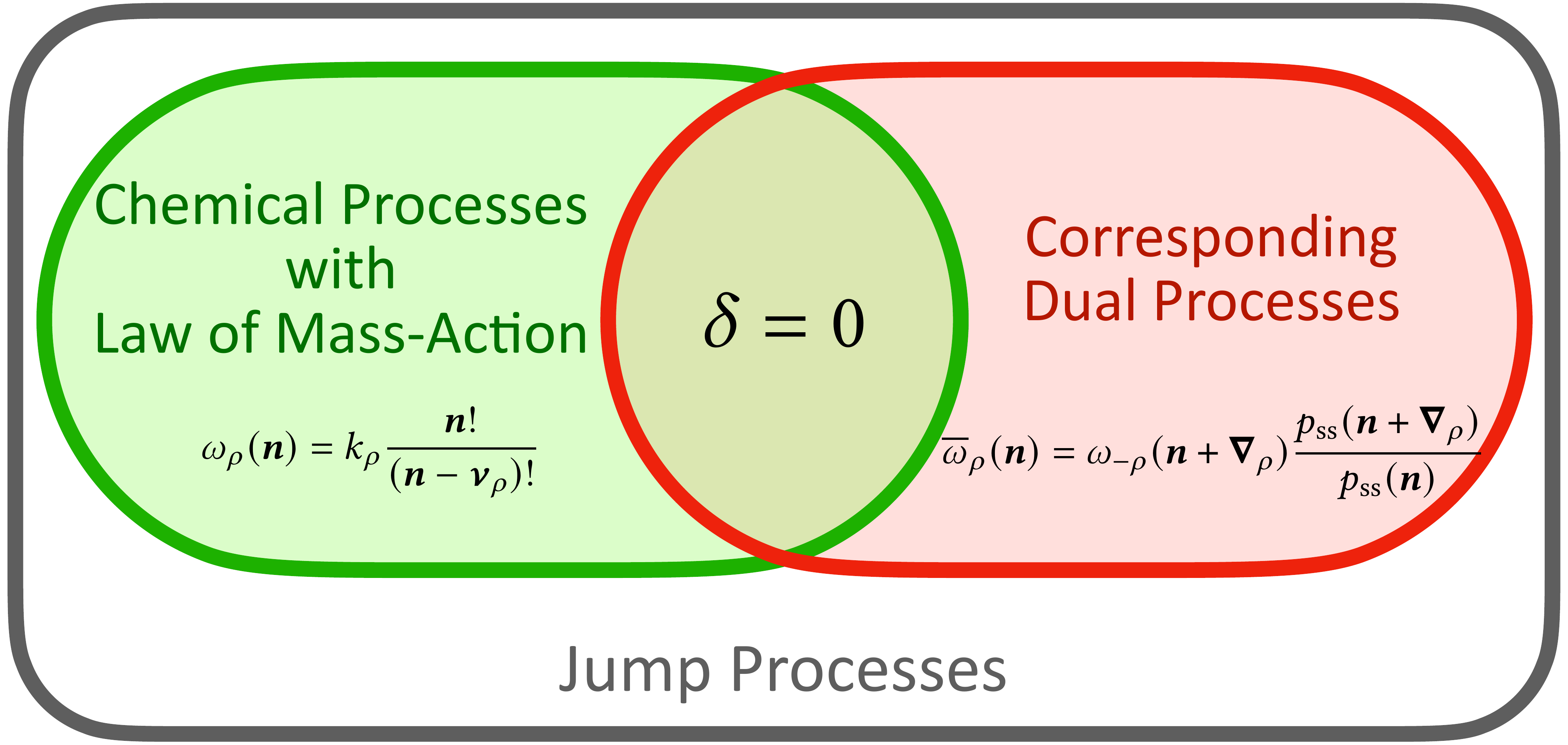}
\caption{
The (gray) box represents the set of all possible Markov jump processes, characterized by different transition rates, on the state space  $\{\n\}$.
One subset (green oval) contains the chemical processes with rates satisfying the law of mass-action. 
Depending on the specific transition rates, the chemical processes represent different CRNs with different deficiency.
The corresponding \dual\ processes are contained in another subset (red oval) and their rates do not satisfy in general the law of mass-action.
{Except for a set of null measure}, 
only for deficiency-zero CRNs, i.e., $\deff =0$, the subset of chemical processes and the subset of the corresponding \dual\ processes intersect, namely, the \dual\ processes are still chemical processes.
}
\label{fig:reversed}
\end{figure}


\begin{theorem}
\label{th:timereversal}
On a reversible CRN endowed with mass-action rates, 
the \dual\ rates satisfy the law of mass-action for all values of the rate constants if and only if the CRN has zero deficiency,
in which case the \dual\ rate constants are given by
\begin{equation}
\rd{k}_{\rct} = k_{-\rct} \, \sscc^{\colS}\,,
\label{eq:rdk}
\end{equation} 
where $\sscc\in\realp^\nspe$ is the fixed point of the corresponding deterministic dynamics
, i.e., $\matS \dcurr(\sscc)= 0$ with $\dcurr(\cc)= (\dots,\dcurre_{\rct}(\cc),\dots)^\intercal_{\rct>0}$ and $ \dcurre_{\rct}(\cc)= k_{\rct} \cc^{\stov{}} - k_{-\rct} \cc^{\stov{-}}$ 
(here we used the following notation $\vvv^\www = \prod_\iii{\vvve_\iii}^{\wwwe_\iii}$ for every pair of vectors $\vvv$ and $\www$).
\end{theorem}

\begin{proof}
We use two previous results (whose derivation is summarized in App.~\ref{sec:ssd0}) connecting the steady-state probability of CRNs and their deficiency.
The first result~\cite{Anderson2010} states that 
if a CRN has zero deficiency, 
then its steady-state probability $\probss(\n)$ on $\stsp$ is given by the Poisson-like distribution
\begin{equation}
\probss(\n) = \frac{1}{\pnorm}\frac{\sscc^{\n}}{\n!}\,,
\label{eq:sspl}
\end{equation}
where $\pnorm$ is the normalization constant over $\stsp$.
Note that $\probss(\n)$ in Eq.~\eqref{eq:sspl} is ``Poisson-like'' and not simply ``Poisson'' because $\stsp$ does not necessarily span all of $\naturals$.
The second results~\cite{Cappelletti2016} states that a generalized converse is also true:
if for any value of the kinetic constants the steady-state probability of a CRN on any stoichiometric subspace $\stsp$ is of the form Eq.~\eqref{eq:sspl}, with parameter $\sscc$ independent of the stoichiometric subspace, then the CRN has zero deficiency.

{\it Sufficiency.} If a CRN has zero deficiency, then it admits the Poisson-like steady-state probability~\eqref{eq:sspl} and Eq.~\eqref{eq:rratema} is recovered from Eq.~\eqref{eq:rrate}. 

{\it Necessity.} If the \dual\ rates satisfy the law of mass-action~\eqref{eq:rratema}, then Eq.~\eqref{eq:rrate} leads to
\begin{equation}
\frac{(\np)!\,\probss(\np)}{\n!\,\probss(\n)}=\frac{\rd{k}_{\rct}}{k_{-\rct}}=:\exp{\hr{}}\,.
\label{eq:defh}
\end{equation}
The right-most equality defines~$\hr{}$ which is independent of the population $\n$, and so of the stoichiometric subspace~$\stsp$.
Furthermore, $\hr{}$ is antisymmetric, i.e., $\hr{-}= -\hr{}$. 
Now, given an arbitrary reference population $\na$, consider 
the succession of reactions $\{\rcttrja,\dots,\rcttrj,\dots,\rcttrjf\}$ such that $\na + \sum_{\indext=1}^{\nstep} \colSs{\indext}= \nt$.
Then, we define $\srcte(\trj{}) := \sum_{\indext=1}^{\nstep} (\dk_{\rct, \rcttrj} - \dk_{-\rct, \rcttrj} ) \in \integer$ as the net number of times reaction $\rct>0$ occurs (negative signs identify reactions occurring in backward direction) along the trajectory $\trj{}$.
By definition, we have
\begin{equation}
\nt - \na =  \sum_{\indext=1}^{\nstep} \colSs{\indext} = \matS \srct(\trj{})\,,
\label{eq:rctcount}
\end{equation}
with $\integer^\nrct\ni\srct(\trj{}) = (\dots, \srcte(\trj{}),\dots )^\intercal_{\rct>0}$.
Taking products of Eq.~\eqref{eq:defh} along the trajectory $\trj{}$ yields
\begin{equation}
\frac{\nt!\,\probss(\nt)}{\na!\,\probss(\na)}=:\exp({\h\cdot \srct(\trj{}) })\,,
\label{eq:rctcount_impl}
\end{equation}
with $\h = (\dots,\hr{},\dots)_{\rct > 0}$.
Now consider a closed trajectory such that $\nt = \na$.
{First, Eq.~\eqref{eq:rctcount} implies $\matS\srct(\trj{0})=0$,
i.e., $\srct(\trj{0})$ is an integer (right) null vector of~$\matS$, i.e, a cycle $\cycle = \srct(\trj{0})$ as defined in Eq.~\eqref{eq:scycle}.
Second, Eq.~\eqref{eq:rctcount_impl} implies 
\begin{equation}
\h\cdot \srct(\trj{0}) = 0\,.
\label{eq:hcurr}
\end{equation}
Since $\h$ is independent of the stoichiometric subspace $\stsp$, Eq.~\eqref{eq:hcurr} must hold for every $\na\in\naturals$.
Thus, we can always make $\srct(\trj{0})$ span the null space of $\matS$ by selecting different closed trajectories in a stoichiometric subspace $\stsp$ where all reactions can occur (namely, where there is no reaction $\rct\in\setrct$ such that $\n<\stov{}$ for all $\n\in\stsp$).
This implies that $\h$ must be orthogonal to all null vectors of $\matS$, that is, it must live in the image of $\matS$:
there exists $\generator  = ( b_1,  b_2, \dots, b_\nspe)^\intercal$ such that
}
\begin{align}
\h = \generator \cdot \matS\,.  
\label{eq:hh}
\end{align}
Note that $\generator$ is also independent of the stoichiometric subspace $\stsp$ since $\h$ is independent of $\stsp$.
Finally, for an open trajectory Eqs.~\eqref{eq:rctcount},~\eqref{eq:rctcount_impl} and~\eqref{eq:hh} yield
\begin{subequations}
\begin{align}
\probss(\nt) 
& =  \frac{\na! \, \probss(\na)}{\nt!}  \exp ( \generator \cdot \matS \srct(\trj{}))\\ 
& =   \frac{\na! \, \probss(\na)}{\nt!}  \, \exp( \generator \cdot (\nt - \na))\\
& =  \frac{\na! \, \probss(\na)}{ \exp( \generator \cdot  \na)} \frac{ \exp( \generator \cdot \nt )}{\nt!}  
\end{align}%
\end{subequations}
which,
upon the identification $\sscc = \exp \generator$
(where we used the notation $\exp\generator = (\exp b_1, \exp b_2,\dots, \exp b_\nspe)^\intercal $)
and ${\pnorm} = (\exp(\generator \cdot \na))/(\na!\, \probss(\na) )$,
is  the Poisson-like distribution given in Eq.~\eqref{eq:sspl}.
Since this is the steady-state probability for all values of the rate constants only if the CRN has zero deficiency, 
then having \dual\ rates satisfying the law of mass-action implies that the CRN has zero deficiency.
\end{proof}

Notice that there could be sets of kinetic constants of null measure such that the \dual\ rates satisfy the law of mass-action, but the CRN has deficiency $\deff>0$ (see App.~\ref{sec:ssd0}).

The physical implication of our theorem is that the dynamics of CRNs with non-zero deficiency cannot be inverted even if the kinetic constants of all chemical reactions are controlled.

\paragraph*{Example.} 
Consider the following CRN
\begin{equation}
\begin{split}
\ch{\spex&<=>[ ${+1}$ ][ ${-1}$ ] 2 \spex }\\
\ch{\spex &<=>[ ${+2}$ ][ ${-2}$ ] $\emptyset$}
\end{split}\,,
\label{eq:excrn2}
\end{equation}
with deficiency $\deff = 1$.
Here, $\nex\in\stsp = \naturaln$ is the number of molecules of species $\ch{\spex}$ 
and, according to Eq.~\eqref{eq:massaction},
\begin{subequations}
\begin{align}
\ratex{+1}(\nex) &= k_{+1}\nex\,, \\
\ratex{-1}(\nex) &= k_{-1}\nex(\nex -1)\,, \\
\ratex{+2}(\nex) &= k_{+2}\nex\,, \\
\ratex{-2}(\nex) &= k_{-2}\,,
\end{align}
\label{eq:exrates}%
\end{subequations}
are the reaction rates. 
We now determine the analytical expression of the \dual\ rates~\eqref{eq:rrate} and we show they satisfy the law of mass-action~\eqref{eq:rratema} only for specific values of the kinetic constants $\{k_{\pm1}, k_{\pm2}\}$ of null measure.

First, by using Eq.~\eqref{eq:ss}, we find that the steady-state probability is given by 
\begin{equation}
\frac{\probss(\nex+1)}{\probss(\nex)} = \frac{k_{-2}+ k_{+1}\nex}{(\nex + 1)(k_{+2}+k_{-1}\nex)}\,,
\label{eq:exss}
\end{equation}
which, together with Eqs.~\eqref{eq:rrate} and~\eqref{eq:exrates}, leads to the following \dual\ rates:
\begin{subequations}
\begin{align}
\rratex{+1}(\nex) &= k_{-1}\nex \frac{k_{-2}+ k_{+1}\nex}{k_{+2}+k_{-1}\nex}\,, \\
\rratex{-1}(\nex+1) &= k_{+1}\nex \frac{(\nex + 1)(k_{+2}+k_{-1}\nex)}{k_{-2}+ k_{+1}\nex}\,, \\
\rratex{+2}(\nex+1) &= k_{-2} \frac{(\nex + 1)(k_{+2}+k_{-1}\nex)}{k_{-2}+ k_{+1}\nex}\,, \\
\rratex{-2}(\nex) &= k_{+2} \frac{k_{-2}+ k_{+1}\nex}{k_{+2}+k_{-1}\nex}\,.
\end{align}
\label{eq:exrrates}%
\end{subequations}
Second, by assuming now that the above \dual\ rates satisfy the law of mass-action according to Eq.~\eqref{eq:rratema}, 
we find that the kinetic constants $\{\rd{k}_{\pm1}, \rd{k}_{\pm2}\}$ read
\begin{subequations}
\begin{align}
\rd{k}_{+1}&= k_{-1} \frac{k_{-2}+ k_{+1}\nex}{k_{+2}+k_{-1}\nex}\,, \\
\rd{k}_{-1} &= k_{+1} \frac{k_{+2}+k_{-1}\nex}{k_{-2}+ k_{+1}\nex}\,, \\
\rd{k}_{+2} &= k_{-2} \frac{k_{+2}+k_{-1}\nex}{k_{-2}+ k_{+1}\nex}\,, \\
\rd{k}_{-2}&= k_{+2} \frac{k_{-2}+ k_{+1}\nex}{k_{+2}+k_{-1}\nex}\,.
\end{align}
\label{eq:exkconstant}%
\end{subequations}
This can only happen if the right hand sides of Eqs.~\eqref{eq:exkconstant} are constant for every $\nex$,
namely, if the kinetic constants $\{k_{\pm1}, k_{\pm2}\}$ satisfy $(k_{-2}+ k_{+1}\nex) = \cost(k_{+2}+k_{-1}\nex)$ $\forall \nex$ with $\cost\in\realp$, or equivalently
\begin{equation}
 k_{-2} = \cost k_{+2}\,,\qquad
 k_{+1} = \cost k_{-1}\,.
 \label{eq:exconditionk}
\end{equation}
This set of kinetic constants is of null measure compared to the set of all possible kinetic constants
and makes the CRN~\eqref{eq:excrn2} detailed balance.

Note that by using Eq.~\eqref{eq:exconditionk} in Eq.~\eqref{eq:exss} and determining the corresponding normalization constant, the steady-state probability~\eqref{eq:exss} becomes a Poisson distribution:
\begin{equation}
\probss(\nex) = e^{-\cost}\frac{\cost^\nex}{\nex!}\,.
\end{equation}

\section{Deficiency in Catalytic CRNs~\label{sec:deficiencycat}} 
According to the International Union of Pure and Applied Chemistry~\cite{iupac2009}, a catalyst is 
``a substance that increases the rate of a reaction without modifying the overall standard Gibbs energy change in the reaction; 
the process is called catalysis. 
The catalyst is both a reactant and product of the reaction.'' 
Furthermore, ``catalysis brought about by one of the products of a net reaction is called autocatalysis.''
From this definition, necessary stoichiometric conditions for catalysis in CRNs have been recently defined in Ref.~\cite{Blokhuis2020},
where also the term allocatalysis has been introduced to refer to the case of standard catalysis and contrast it with the case of autocatalysis.
Here, we build on these stoichiometric conditions to show that catalytic CRNs have deficiency $\deff>0$ when driven by exchange reactions (see Subs.~\ref{sub:engine}).

\subsection{Stoichiometric Conditions for Catalytic CRNs~\label{subs:defcat}} 
We reformulate here the stoichiometric conditions for catalysis, given in Ref.~\cite{Blokhuis2020}, for the framework introduced in Sec.~\ref{sec:crn}.
A CRN is said to be \text{catalytic} if 
the set of chemical species $\setspe$ can be split into two disjoint subsets $\setspe = \setspeQ \cup \setspeP$ 
such that the (sub)stoichiometric matrix for the internal reactions can be written as
\begin{equation}
\renewcommand*{\arraystretch}{1.4}
\matSin = 
	\begin{pmatrix}
		\matSinQ\\
		\matSinP
	\end{pmatrix} \, ,
\label{eq:matSinQP}
\end{equation}
where i) $\matSinQ$ is autonomous (i.e., each column of $\matSinQ$ has at least one positive and one negative entry), and
ii) there is a reaction vector $\srctc$ satisfying 
\begin{subequations}
\begin{align}
&\matSinQ \srctc \geq 0\,,\label{eq:cond_cat}\\
&\matSinP \srctc \neq 0\,,\label{eq:cond_sub}
\end{align} \label{eq:cond}%
\end{subequations}
where the inequality in Eq.~\eqref{eq:cond_cat}, as well as any inequality involving vectors in the following, must hold entry-wise.
The $\nspeq$ species $\{\chem\}_{\spe\in\setspeQ}$ are the catalysts, while the $\nspep$ species $\{\chem\}_{\spe\in\setspeP}$ are the substrates (with $\nspeq + \nspep = \nspe$).
The condition of autonomy ensures that the production of each catalyst in the internal reactions is conditioned on the presence of other catalysts
in agreement with the definition given in Ref.~\cite{iupac2009}: a ``catalyst is both a reactant and product of the reaction''.
Notice that the condition of autonomy is not restrictive as catalysts can be exchanged with the environment (via reaction~\eqref{eq:exrct}), where they can be produced/consumed by other chemical reactions.

\paragraph*{Remark.} 
In Ref.~\cite{Blokhuis2020}, $\matSinQ$ is also assumed to be non-ambiguous: each species can participate in a chemical reaction as either a reagent or a product.
Reactions like $\ch{F + E <=> E + W}$ or $\ch{F + E <=> 2 E}$ (where $\ch{E}$ participate as both a reagent and a product) are therefore excluded.
Non-ambiguity is essential to the decomposition (done in Ref.~\cite{Blokhuis2020}) of autocatalytic CRNs into minimal autocatalytic subnetworks, named autocatalytic motifs or cores. 
In our analysis, non-ambiguity is not required.

The general conditions for catalysis in Eq.~\eqref{eq:cond} can be further specialized for allocatalysis and autocatalysis.
On the one hand, if 
\begin{equation}
\matSinQ \srctc = 0\label{eq:cond_allo}
\end{equation}
the CRN is said to be allocatalytic and the species $\{\chem\}_{\spe\in\setspeQ}$ are named allocatalysts.
In this case, $\srctc$ physically defines the number of times each internal reaction occurs along an allocatalytic cycle, 
namely, a sequence of reactions that upon completion leaves the abundances of the allocatalysts unchanged (Eq.~\eqref{eq:cond_allo}), while interconverting the substrates via an effective reaction whose stoichiometry is specified by $\matSinP \srctc$ in Eq.~\eqref{eq:cond_sub}.
On the other hand, if 
\begin{equation}
\matSinQ \srctc > 0\label{eq:cond_auto}\,,
\end{equation}
the CRN is said to be autocatalytic and the species $\{\chem\}_{\spe\in\setspeQ}$ are named autocatalysts.
In this case, $\srctc$ physically defines the number of times each internal reaction occurs to produce a net increase in the abundances of the autocatalysts (Eq.~\eqref{eq:cond_auto}), while interconverting the substrates according to the stoichiometry specified by $\matSinP \srctc$ in Eq.~\eqref{eq:cond_sub}.

\paragraph*{Example.}  
In the CRN~\eqref{eq:excrn}, one can classify $\{\ch{I}\,,\ch{E}\}$ and $\{\ch{W}\,,\ch{F}\}$ as the sets of catalysts and substrates, respectively.
Thus, the stoichiometric matrix~\eqref{eq:exmatS} can be split in the following submatrices 
\begin{equation}
\matSinQ=
 \kbordermatrix{
    & \color{g}1 &\color{g}2		  \cr
    \color{g}\ch{I}    	& 1 &-1 	  \cr
    \color{g}\ch{E}   	&-a & (a+b) \cr
  }\,,\text{ }\text{ }\text{ }\text{ }\text{ }
\matSinP=
 \kbordermatrix{
    & \color{g}1 &\color{g}2		  \cr
    \color{g}\ch{W}  	& 0 & 1 	  \cr
    \color{g}\ch{F}   	&-1 & 0 	  \cr
  }\,
\label{eq:exmatSQP} 
\end{equation}
and
\begin{equation}
\matSex=
 \kbordermatrix{
    & \color{g}3 &&\color{g}4&\color{g}5\cr
    \color{g}\ch{I}   	& 0 &\vrule& 0 & 0\cr
    \color{g}\ch{E}   	&-1 &\vrule& 0 & 0\cr\cline{2-5}
    \color{g}\ch{W}  	& 0 &\vrule&-1 & 0\cr
    \color{g}\ch{F}   	& 0 &\vrule& 0 &-1\cr
  }\,,
  \label{eq:exmatSex}
\end{equation}
where $\matSinQ$ is autonomous for $a> 0$, and vertical and horizontal lines in $\matSex$ in Eq.~\eqref{eq:exmatSex} are explained in Subs.~\ref{subs:catce}.

On the one hand, if $a>0$ and $b=0$, the CRN~\eqref{eq:excrn} is allocatalytic.
Indeed, we can identify the reaction vector $\srctc =(1\,,1)^\intercal$ such that $\matSinQ \srctc = 0$ and $\matSinP\srctc = (1\,, -1)^\intercal\neq 0$, 
while there is no $\srctc$ such that $\matSinQ \srctc > 0$.
Along the allocatalytic cycle $\srctc$, 
the abundance of the autocatalysts remains constant and substrates are interconverted by the effective chemical reaction
\begin{equation}
\ch{F <=>[ ${\srctc}$ ][ ${}$ ] W}\,.
\end{equation}

On the other hand, if $a>0$ and $b>0$, the CRN~\eqref{eq:excrn} is autocatalytic.
Indeed, we can identify the reaction vector $\srctc =\big((a + b +1)\,, (a+1)\big)^\intercal$ such that $\matSinQ \srctc = (b\,,b)^\intercal>0$
and $\matSinP\srctc = \big((a+1)\,, -(a+b+1)\big)^\intercal\neq 0$, 
while there is no $\srctc$ such that $\matSinQ \srctc = 0$ since $\matSinQ$ is full rank.
Along the autocatalytic sequence of reactions  $\srctc$, 
the abundance of the autocatalysts increases by $b$ and substrates are interconverted according to the stoichiometry
\begin{equation}
\ch{$(a+b+1)$ F <=>[ ${\srctc}$ ][ ${}$ ] $(a+1)$ W}\,.
\end{equation}

\subsection{Driven Catalytic CRNs~\label{subs:catce}}
We specialize here the notion of driven CRNs, defined in Subs.~\ref{sub:engine}, to catalytic CRNs.
To do so, we recognize that, from a chemical point of view~\cite{iupac2009}, catalytic CRNs must primarily interconvert substrates since the ``catalyst is both a reactant and product of the reactions''.
Thus, we assume that a catalytic CRN is driven 
when \textit{at least} some substrates are exchanged with the environment. 
Hence, the stoichiometric cycle~\eqref{eq:subcycle} must satisfy 
\begin{equation}
\cycleExP \neq 0\,,
\label{eq:couplingEnvP}
\end{equation} 
when written as
\begin{equation}
\cycle = \big(\cycleIn\,,\cycleExQ\,,\cycleExP\big)^\intercal
\label{eq:subcycleQP}
\end{equation}
by applying the splitting $\setrctEx = \setrctExQ\cup\setrctExP$, with $\setrctExQ$ (resp. $\setrctExP$) 
the set of $\nspeqex$ (resp. $\nspepex$) pairs of reactions exchanging some catalysts (resp. substrates)

Note that for driven catalytic CRNs the (sub)stoichiometric matrix for the exchange reactions in Eq.~\eqref{eq:matSinex} specializes into
\begin{equation}
\renewcommand*{\arraystretch}{1.4}
\matSex = 
	\begin{pmatrix}
		\matSexQ\,,		&\matOexQ\\
		\matOexP\,,		&\matSexP
	\end{pmatrix} \,,
\label{eq:matSinexQP} 
\end{equation}
where $\matSexQ$ (resp. $\matSexP$) is the $\nspeq\times\nspeqex$ (reps. $\nspep\times\nspepex$) matrix
whose columns $\setrctExQ\ni\rcts>0$ (resp. $\setrctExP\ni\rcts>0$) have all null entries 
except for the row corresponding to the $\spe$ catalyst (resp. $\spe$ substrate) that is equal to~$-1$;
$\matOexQ$ (resp. $\matOexP$) is the $\nspeq\times\nspepex$ (resp. $\nspep\times\nspeqex$) null matrix resulting from the fact that no catalysts (resp. substrates) are involved in the exchange reactions $\setrctExP$ of the substrates (resp. $\setrctExQ$ of the catalysts).

Furthermore, Eq.~\eqref{eq:cycleimplication} specializes into 
\begin{equation}
\matSexP\cycleExP\neq0\,.
\label{eq:cycleimplicationP}
\end{equation}
The condition in Eq.~\eqref{eq:couplingEnvP} together with Eq.~\eqref{eq:cycleimplicationP} also implies  $\cycleIn\neq0$
since $\matS\cycle=0$ requires $\matSinP\cycleIn = -\matSexP\cycleExP\neq0$ (because of Eqs.~\eqref{eq:matSinexQP} and~\eqref{eq:cycleimplicationP}).
Therefore,  $\cycleExP\neq0$ (resp. $\cycleIn\neq0$) 
ensures that the net steady-state current $\currss{}(\n)$ of the exchange (resp. internal) reactions $\setrctExP\ni\rct>0$ (resp. $\setrctIn\ni\rct>0$) such that $\cyclee\neq0$ does not vanish.
This physically means that (some) substrates are continuously interconverted by the catalysts in the internal reactions and continuously exchanged with the environment at steady state.

\paragraph*{Example.}  
We split the exchange reactions $\setrctEx$ of the CRN~\eqref{eq:excrn} (introduced in Subs.~\ref{sub:engine}) into
the exchange reactions of the catalyst $\ch{E}$ $\setrctExQ = \{\pm3\}$
and the exchange reactions of the substrates $\{\ch{W}\,,\ch{F}\}$ $\setrctExP = \{\pm4\,,\pm5\}$.
By applying the same splitting to the stoichiometric cycle~\eqref{eq:excycle},
\begin{equation}
\cycle =(\underbrace{1,1}_{=\cycleIn},\underbrace{b}_{=\cycleExQ},\underbrace{1,-1}_{=\cycleExP})^\intercal\,,
\label{eq:excycle2}
\end{equation}
with $\cycleExP\neq 0$ for every $a\geq0$ and $b\geq0$,
which physically means that the substrates $\{\ch{W}\,,\ch{F}\}$ are always exchanged with the environment. 
Hence, the CRN~\eqref{eq:excrn} is a driven catalytic CRN if $a>0$ and $b\geq0$ (see Subs.~\ref{subs:defcat}). 

Furthermore, the vertical and horizontal lines in Eq.~\eqref{eq:exmatSex} split the matrix $\matSex$ into the blocks introduced in Eq.~\eqref{eq:matSinexQP}.
This shows that $\matOexQ$ (resp. $\matOexP$) is the null matrix accounting for the fact that 
the catalysts $\{\ch{I},\ch{E}\}$ (resp. the substrates $\{\ch{W},\ch{F}\}$) are not involved 
in the exchange reactions of the substrates $\setrctExP = \{\pm4\,,\pm5\}$ (resp. of the catalysts $\setrctExQ = \{\pm3\}$).

\paragraph*{Remark.}
In some cases the stoichiometric cycle $\cycle$ in Eq.~\eqref{eq:subcycleQP} can be expressed in terms of the catalytic vector $\srctc$.

When all substrates of allocatalytic CRNs are involved in the exchange reactions, 
the stoichiometric matrix admits the stoichiometric cycle with
\begin{subequations}
\begin{align}
\cycleIn &= \srctc \,,\\
\cycleExQ &= \boldsymbol 0 \,,\\
\cycleExP &= \matSinP \srctc\,,
\end{align}
\label{eq:scycle_allo}%
\end{subequations}
satisfying Eq.~\eqref{eq:couplingEnvP}  because of the stoichiometric condition for allocatalysis given in Eqs.~\eqref{eq:cond_sub} and~\eqref{eq:cond_allo}.
Note that $\cycle$ in Eq.~\eqref{eq:scycle_allo} is a cycle of the stoichiometric matrix 
because $-\matSexP$ becomes, without loss of generality, the $\nspep\times\nspep$ identity matrix. 
Nevertheless, exchanging a subset of the substrates with the environment might be in general sufficient for the emergence of a stoichiometric cycle satisfying Eq.~\eqref{eq:couplingEnvP} in allocatalytic CRNs. 

When all substrates and autocatalysts of autocatalytic CRNs are involved in the exchange reactions, 
the stoichiometric matrix admits the stoichiometric cycle with
\begin{subequations}
\begin{align}
\cycleIn &= \srctc \,,\\
\cycleExQ &= \matSinQ \srctc \,,\\
\cycleExP &= \matSinP \srctc\,,
\end{align}
\label{eq:scycle_auto}%
\end{subequations}
satisfying Eq.~\eqref{eq:couplingEnvP}  because of the stoichiometric condition for autocatalysis given in Eqs.~\eqref{eq:cond_sub} and~\eqref{eq:cond_auto}.
Note that $\cycle$ in Eq.~\eqref{eq:scycle_auto} is a cycle of the stoichiometric matrix
because also $-\matSexQ$ becomes, without loss of generality, the $\nspeq\times\nspeq$ identity matrix. 
Nevertheless, exchanging a subset of the autocatalysts and substrates with the environment might be in general sufficient for the emergence of a stoichiometric cycle satisfying Eq.~\eqref{eq:couplingEnvP} in autocatalytic CRNs.

\subsection{Deficiency in Catalytic CRNs\label{subs:def_cat}} 
We show now that driven catalytic CRNs are not deficiency-zero.

\begin{theorem}
\label{th:deficienycat}
A catalytic CRN (as defined in Subs.~\eqref{subs:defcat}) 
driven out of equilibrium via exchange reactions (as defined in Subs~\eqref{subs:catce}) has deficiency~$\deff>0$.
\end{theorem}

\begin{proof}
By definition, the (sub)stoichiometric matrix $\matSinQ$ of catalytic CRNs is autonomous (Subs.~\eqref{subs:defcat}).
This implies that,
while the exchanged substrates (together with the environment $\emptyset$) are the complexes of their exchange reactions, 
the internal reactions interconvert different complexes since they always involve some catalysts.
Thus, the submatrix of the incidence matrix $\matI$ that collects the rows corresponding to the substrates can be written as 
\begin{equation}
(\matOinP\,,\matIexP )\,,
\label{eq:submatIP}
\end{equation}
where
$\matIexP$ is the incidence matrix of the substrates in the exchange reactions, and
$\matOinP$ is the $\nspep\times\nrctin$ (with $\nrctin$ the number of pairs of internal reactions) null matrix resulting from the fact that substrates are not complexes of the internal reactions. 
Since some exchange reactions might involve catalysts, $\matIexP$ reads
\begin{equation}
\matIexP = (\matOexP\,,\matSexP)\,,
\label{eq:submatIP2}
\end{equation}
with $\matOexP$ and $\matSexP$ already introduced in Eq.~\eqref{eq:matSinexQP}.
Because of the submatrix~\eqref{eq:submatIP}, 
the stoichiometric cycle $\cycle$ 
satisfying Eq.~\eqref{eq:couplingEnvP}, whose existence is granted by definition of driven catalytic CRNs (Subs~\eqref{subs:catce}),
is not a topological cycle: $\matI\cycle\neq0$.
Indeed, by using Eqs.~\eqref{eq:submatIP2},~\eqref{eq:subcycleQP} and~\eqref{eq:cycleimplicationP} 
\begin{equation}
(\matOinP\,,\matIexP )\, \cycle = \matSexP\cycleExP\neq0\,.
\label{eq:proofdeficiency}
\end{equation}
Hence, driven catalytic CRNs admit at least one stoichiometric cycle that is not a topological cycle and $\deff>0$.
\end{proof}


\paragraph*{Example.}  
As already discussed, the CRN~\eqref{eq:excrn} is driven by the exchange reactions (Sec.~\ref{sec:crn}) and is catalytic when $a>0$ and $b\geq0$ (Subs.~\ref{subs:defcat}). 
In particular, it is allocatalytic when $a>0$ and $b=0$, while it is autocatalytic when $a>0$ and $b>0$.
Notice that the existence of the stoichiometric cycle~\eqref{eq:excycle} is granted also for the autocatalytic case even if only \ch{E} is exchanged with the environment.
By labeling now the complexes as $\comps{1}=\ch{I}$, $\comps{2}=\ch{F + $a$E}$, $\comps{3}=(a + b)\ch{E + W}$, $\comps{4}=\ch{E}$, $\comps{5}=\ch{W}$ and $\comps{6}=\ch{F}$, $\comps{7}=\emptyset$, the incidence matrix reads
\begin{equation}
\matI=
 \kbordermatrix{
    & \color{g}1 &\color{g}2&&\color{g}3&&\color{g}4&\color{g}5\cr
    \color{g}\comps{1}   	& 1 &-1 & \vrule & 0 & \vrule  & 0 & 0 \cr
    \color{g}\comps{2} 	&-1 & 0 & \vrule & 0 & \vrule  & 0 & 0 \cr
    \color{g}\comps{3}		& 0 & 1 & \vrule & 0 & \vrule  & 0 & 0 \cr
    \color{g}\comps{4}		& 0 & 0 & \vrule &-1 & \vrule  & 0 & 0 \cr\cline{2-8}
    \color{g}\comps{5}		& 0 & 0 & \vrule & 0 & \vrule  &-1 & 0 \cr
    \color{g}\comps{6}		& 0 & 0 & \vrule & 0 & \vrule  & 0 &-1 \cr\cline{2-8}
    \color{g}\comps{7}		& 0 & 0 & \vrule & 1 & \vrule  & 1 & 1 \cr
  }\,,\label{eq:exmatI_cat}
\end{equation}
where the block introduced in Eq.~\eqref{eq:submatIP} is framed between the horizontal lines and split into $(\matOinP\,,\matOexP\,,\matSexP )$ by the vertical lines. 
Since $\cycle$ in Eq.~\eqref{eq:excycle} is a right null eigenvector of the stoichiometric matrix~\eqref{eq:exmatS}, but not of the incidence matrix~\eqref{eq:exmatI_cat}, the catalytic version of the CRN~\eqref{eq:excrn} has deficiency $\deff>0$.

\section{Discussion and Conclusions\label{sec:discussion}}
In this work, we showed that the dynamics of driven catalytic CRNs cannot be inverted
by controlling the kinetic constants~\footnote{
The kinetic non-invertibility of catalytic CRNs implies that there are no kinetic constants that can generate the dual process, 
where the net steady-state currents of \textit{all} reactions are inverted.
However, this does not exclude that the net steady-state currents of \textit{some} specific reactions 
(e.g., the reaction producing a desired species) can be inverted in catalytic CRNs by controlling the kinetic constants}
as summarized in Fig.~\ref{fig:mainresults}.
This conclusion is reached by combining our two main results.
\begin{figure}[t]
  \centering
  \includegraphics[width=0.99\columnwidth]{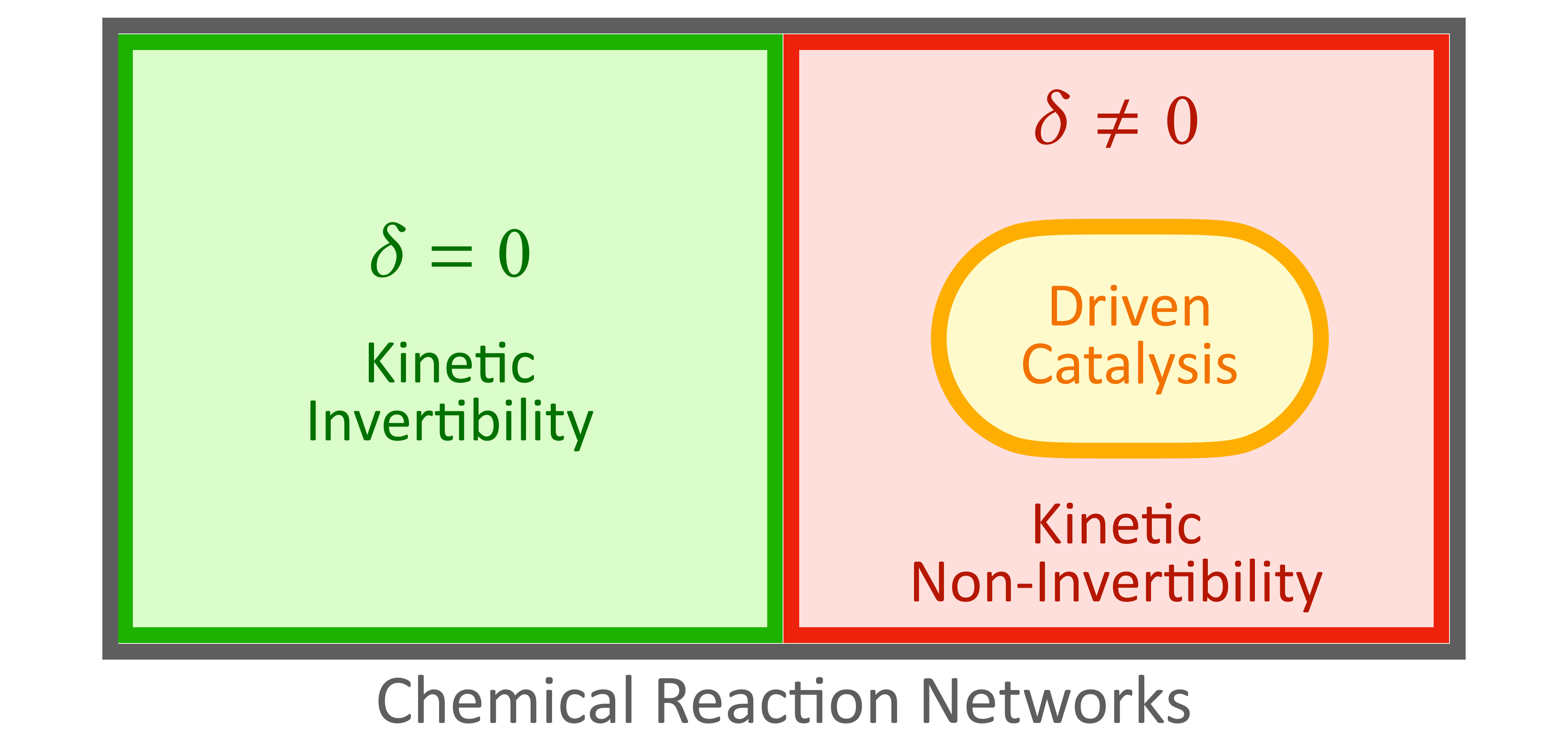}
\caption{Set of all possible CRNs (gray box) can be split into two subsets.
One subset (green box) includes the CRNs with deficiency $\deff = 0$,
whose \dual\ dynamics is still a chemical process satisfying the law of mass-action. 
The dynamics of these CRNs is kinetically invertible.
The other subset (red box) includes the CRNs with deficiency $\deff \neq 0$,
whose \dual\ dynamics is not a chemical process satisfying the law of mass-action. 
The dynamics of these CRNs is kinetically non-invertible.
Driven catalytic CRNs constitute a subset of the kinetically non-invertible CRNs.
}
\label{fig:mainresults}
\end{figure}

Our first  main result (Sec.~\ref{sec:chem_rev}) is that the \dual\ master equation of a reversible stochastic CRN describes a chemical process satisfying mass-action law if and only if the CRN has zero deficiency, i.e., $\deff=0$.
On the one hand, this implies that the dual process of CRNs with $\deff = 0$
can be (in principle) physically realized by setting the value of the kinetic constants according to Eq.~\eqref{eq:rdk}.
Hence, the stochastic dynamics of CRNs with $\deff = 0$ is kinetically invertible.
On the other hand,
for all CRNs with $\deff>0$, the dual process cannot be physically realized:
there are no values of the kinetic constants (or, equivalently, mass-action rates) 
that can invert the dynamics (either in terms of net steady-state currents as discussed in Sec.~\ref{sec:dynamics} or in terms of stochastic trajectories as discussed in App.~\ref{app:trarep}).
This implies that the stochastic dynamics of CRNs with $\deff>0$ is kinetically non-invertible 
as a result of their topology.
However, their deterministic dynamics can always be inverted independently of the deficiency.
Indeed, the net steady-state current of the deterministic dynamics,
\begin{equation}
\dcurre_{\rct}(\sscc)= k_{\rct} \sscc^{\stov{}} - k_{-\rct} \sscc^{\stov{-}}\,,
\end{equation}
can be inverted by choosing the kinetic constants according to Eq.~\eqref{eq:rdk}:
\begin{equation}
\rdcurre_{\rct}(\sscc)= \rd{k}_{\rct} \sscc^{\stov{}} - \rd{k}_{-\rct} \sscc^{\stov{-}} = - \dcurre_{\rct}(\sscc)\,.
\end{equation}
We leave the investigation of the disappearance of the non-invertibility in the thermodynamic limit~\cite{vankampen} to further studies.

Our second main result  (Sec.~\ref{sec:deficiencycat}) is that driven catalytic CRNs have deficiency~$\deff>0$.
We examine now how two assumptions that we used to prove this second result can be relaxed.
First, the proof in Subs.~\ref{subs:def_cat} still holds for generic exchange reactions (i.e., not satisfying Eq.~\eqref{eq:exrct}) that do not involve both substrates and catalysts together.
Indeed, it only requires that the complexes involved in the exchange reactions are different from those involved in the internal reactions.
Second, the proof in Subs.~\ref{subs:def_cat}, and in particular Eq.~\eqref{eq:proofdeficiency}, still holds if there were internal reactions interconverting only the substrates as long as there is at least a cycle~\eqref{eq:subcycleQP} involving the catalytic reactions.

We conclude by noticing that while driven catalytic CRNs have deficiency~$\deff>0$, driven CRNs with~$\deff>0$ might not be catalytic.
Consider for instance the following CRN
\begin{equation}
\begin{split}
\ch{A&<=>[ ${+1}$ ][ ${-1}$ ] C + C' }\\
\ch{C&<=>[ ${+2}$ ][ ${-2}$ ] D}\\
\ch{C'&<=>[ ${+3}$ ][ ${-3}$ ] D'}\\
\ch{D + D'&<=>[ ${+4}$ ][ ${-4}$ ] B}\\
\ch{A &<=>[ ${+5}$ ][ ${-5}$ ] $\emptyset$}\\
\ch{B &<=>[ ${+6}$ ][ ${-6}$ ] $\emptyset$}
\end{split}\,,
\label{eq:excrn3}
\end{equation}
which admits the stoichiometric cycle $\cycle =(1,1,1,1,1,-1)^\intercal$.
The corresponding incidence matrix admits no cycle and thus the CRN~\eqref{eq:excrn3} has deficiency $\deff = 1$.
However, there is no splitting of the species $\{\ch{A}, \ch{B}, \ch{C}, \ch{C'}, \ch{D}, \ch{D'}\}$ into catalysts and substrates such that $\matSinQ$ is autonomous and Eq.~\eqref{eq:cond} is satisfied.

\begin{acknowledgments}
This research was supported by the Luxembourg National Research Fund (FNR) via
the research funding scheme PRIDE (19/14063202/ACTIVE), 
and the CORE projects ThermoComp (C17/MS/11696700) and ChemComplex (C21/MS/16356329).
M.P. thanks Daniele Cappelletti for fruitful discussions.
\end{acknowledgments}



\appendix
\section{Stochastic Trajectories and Path Probability \label{app:trarep}}
A stochastic trajectory $\trj{}$ of duration $t$ of a chemical process is a sequence of populations $\{\na,\n_1,\dots,\nl,\dots,\nf\}$
generated by a set of reactions $\{\rcttrja,\dots,\rcttrj,\dots,\rcttrjf\}$ 
sequentially occurring at times $\{\ttrja,\dots,\ttrj,\dots,\ttrjf\}$ 
starting from $\na$ at time $t_0=0$ and arriving at $\nf$ at time $t_{\nstep}$. 
Its path probability, conditioned on the initial population, is given by
\begin{equation}
\probtrj(\trj{}) = 
\prod_{\indext=0}^{\nstep}e^{-(\ttrjp - \ttrj)\sum_{\rct\in\setrct}\,\rate{}(\nl)}
\prod_{\indext=1}^{\nstep}\ratel{}(\nlm)
\label{eq:probtrj}
\end{equation}
with $t_{\nstep+1} = t$ and $\nfg = \n$.
The first term in Eq.~\eqref{eq:probtrj} accounts for the probability of dwelling in the state $\nl$ during the time interval $[{\ttrj},{\ttrjp})$,
while second term accounts for the probability that reaction $\rcttrj$ occurs at time $\ttrj$ while being in state $\nlm$.
Since reactions are assumed to be reversible, every stochastic trajectory $\trj{}$ has its time-reversed counterpart $\rtrj{}$.
The time-reversed trajectory is 
a stochastic trajectory of the same duration $t$ 
given by the sequence of populations $\{\nf,\dots,\nl, \dots,\n_1,\na\}$
generated by the set of reactions $\{-\rcttrjf,\dots,-\rcttrj,\dots,-\rcttrja\}$ 
sequentially occurring at times $\{t-\ttrjf,\dots,t-\ttrj,\dots,t-\ttrja\}$
starting from $\nf$ at time $t_0=0$ and arriving at $\na$ at time $t-\ttrja$.
Its path probability, conditioned on the final population, is given by
\begin{equation}
\probtrj(\rtrj{}) = 
\prod_{\indext=0}^{\nstep}e^{-(\ttrjp - \ttrj)\sum_{\rct\in\setrct}\,\rate{}(\nl)} 
\prod_{\indext=1}^{\nstep}\ratel{-}(\nl)\,.
\label{eq:rprobtrj}
\end{equation}
In general, the steady-state probability of a trajectory $\probtrj(\trj{})\probss(\na)$ is different from its time-reversed counterpart $\probtrj(\rtrj{})\probss(\nt)$.
Only when the detailed balance condition~\eqref{eq:db} is satisfied, 
$\probtrj(\trj{})\probeq(\na) = \probtrj(\rtrj{})\probeq(\nt)$.
When the detailed balance condition~\eqref{eq:db} is not satisfied, without loss of generality,
$\probtrj(\trj{})\probss(\na) > \probtrj(\rtrj{})\probss(\nt)$ implying that  the population $\n$ evolves preferentially along the forward trajectory $\trj{}$ rather than along its time-reversed counterpart $\rtrj{}$.
For this reason, chemical processes are said to be antisymmetric under time reversal.

For a given a trajectory $\trj{}$ of the original chemical process, the path probability of the time-reversed trajectory $\rtrj{}$ occurring in the \dual\ process $\rprobtrj(\rtrj{})$ satisfies
\begin{equation}
\probtrj(\trj{}) \probss(\na) = \rprobtrj(\rtrj{}) \probss(\nt) \,.
\end{equation}
This can be easily proven using
$\sum_{\rct\in\setrct}\rrate{}(\n) =\sum_{\rct\in\setrct}\rate{}(\n) $, 
$\prod_{\indext=1}^{\nstep}\rratel{-}(\nl) \probss(\nt)= \prod_{\indext=1}^{\nstep}\ratel{}(\nlm)\probss(\na)$, 
and Eq.~\eqref{eq:rprobtrj}.
This physically means that the time evolution of the population $\n$ in the \dual\ process is inverted compared to the original chemical process.
Indeed, if $\probtrj(\trj{})\probss(\na) > \probtrj(\rtrj{})\probss(\nt)$ then $\rprobtrj(\trj{})\probss(\na) < \rprobtrj(\rtrj{})\probss(\nt)$ implying that in the \dual\ process the population $\n$ evolves preferentially along the time-reversed trajectory $\rtrj{}$ rather than along its forward counterpart $\trj{}$.

\section{Steady-State Probability of Deficiency-Zero CRNs\label{sec:ssd0}}
First, we show that  if a CRN has zero deficiency, 
then the probability in Eq.~\eqref{eq:sspl} satisfies the steady-state condition~\eqref{eq:ss} as proven in Ref.~\cite{Anderson2010}.
We start by recognizing that deficiency-zero CRNs are \textit{complex balanced} for every value of the kinetic constants: 
the fixed point of the deterministic dynamics, i.e., $\matS \dcurr(\sscc)= 0$, must satisfy $\matI\dcurr(\sscc)=0$, which entry-wise reads
\begin{equation}
\sum_{\rct>0} (\dk_{\com,\com(-\rct)} - \dk_{\com,\com(\rct)})\dcurre_{\rct}(\sscc)=0\text{ }\text{ }\text{ }\text{ }\forall\com
\label{eq:detssc}
\end{equation}
since any (right) null eigenvector of the stoichiometric matrix $\matS$ must be a (right) null eigenvector of the incidence matrix $\matI$.
As specified in the main text, $\sscc\in\realp^\nspe$,
$\dcurr(\cc)= (\dots,\dcurre_{\rct}(\cc),\dots)^\intercal_{\rct>0}$,
$ \dcurre_{\rct}(\cc)= k_{\rct} \cc^{\stov{}} - k_{-\rct} \cc^{\stov{-}}$,
and $\vvv^\www = \prod_\iii{\vvve_\iii}^{\wwwe_\iii}$ for every pair of vectors $\vvv$ and $\www$. 
By using $\dcurre_{\rct}(\cc) = -\dcurre_{-\rct}(\cc)$ and defining $\stovcf{} := \stov{}$, Eq.~\eqref{eq:detssc} becomes
$\sscc^{\stovc{}}\sum_{\rct\in\setrct}\dk_{\com,\com(\rct)}(k_{\rct} - k_{-\rct} \sscc^{\colS})=0$,
which implies
\begin{equation}
\sum_{\rct\in\setrct}\dk_{\com,\com(\rct)}\Big(k_{\rct} - k_{-\rct} \sscc^{\colS}\Big)=0\,.
\label{eq:detssc_final}
\end{equation}
By rewriting now the right hand side of the chemical master equation~\eqref{eq:cme} as
\begin{equation}
\sum_{\com}\sum_{\rct\in\setrct}\dk_{\com,\com(\rct)}\{\rate{-}(\np)  \prob(\np) - \rate{}(\n) \prob(\n) \}\,,
\label{eq:ss_complex1}
\end{equation}
plugging the probability given in Eq.~\eqref{eq:sspl} and using $\stovcf{} = \stov{}$, we obtain
\begin{equation}
\sum_{\com}\frac{\sscc^{\n}}{\pnorm(\n-\stovc{})!}\sum_{\rct\in\setrct}\dk_{\com,\com(\rct)}\Big\{  k_{-\rct} \sscc^{\colS} - k_{\rct} \Big\}\,,
\end{equation}
which vanishes because of~\eqref{eq:detssc_final}.
This proves that if a CRN has zero deficiency then the probability in Eq.~\eqref{eq:sspl} satisfies the steady-state condition~\eqref{eq:ss}.

Second, we show that if the probability in Eq.~\eqref{eq:sspl} satisfies the steady-state condition~\eqref{eq:ss}, then the CRN has zero deficiency as proven in Ref.~\cite{Cappelletti2016}.
We start by rewriting the steady-state condition~\eqref{eq:ss} as 
\begin{equation}\small
\sum_{\com}\sum_{\rct\in\setrct}\dk_{\com,\com(\rct)}\Big\{\rate{-}(\np)  \frac{\probss(\np)}{\probss(\n)} - \rate{}(\n) \Big\}=0\,.
\label{eq:ss_complex2}
\end{equation}
By assuming that $\probss(\n)$ is given in Eq.~\eqref{eq:sspl} and using $\stovcf{} = \stov{}$, we obtain 
\begin{equation}
\sum_{\com}\frac{\n!}{(\n-\stovc{})!}\sum_{\rct\in\setrct}\dk_{\com,\com(\rct)}\Big\{  k_{-\rct} \sscc^{\colS} - k_{\rct} \Big\}=0\,.
\label{eq:cappelletti}
\end{equation}
The terms $\{{\n!}/{(\n-\stovc{})!}\}_{\com}$ are linearly-independent polynomials (indeed, the monomial with maximal degree in  ${\n!}/{(\n-\stovc{})!}$ is $\n^{\stovc{}}$ and differ for all complexes $\com$). 
For this reason, Eq.~\eqref{eq:cappelletti} implies that Eq.~\eqref{eq:detssc_final} must be satisfied.
Hence, $\sscc$ is the fixed point of the deterministic dynamics of a complex-balanced CRN.
Complex-balanced CRNs, except for a set of kinetic constants of null measure, have zero deficiency.~\cite{Polettini2015, Cappelletti2016}

\bibliography{biblio}
\end{document}